\documentclass[11pt,a4paper]{article}
\usepackage{fullpage}
\usepackage{microtype}%if unwanted, comment out or use option "draft"
\usepackage{amssymb,amsmath,amsthm} \usepackage{graphicx}
\usepackage{color} \usepackage{xspace}
\usepackage[numbers]{natbib} \usepackage{hyperref,color}
\usepackage{xspace} 
\usepackage[T1]{fontenc}
\graphicspath{{figures/}}
\newtheorem{theorem}{Theorem} 
\newtheorem{lemma}[theorem]{Lemma}

\newtheorem{corollary}[theorem]{Corollary}

\newcommand{\oset}{\mathcal{O}}
\newcommand{\otuple}{\ensuremath{\{o_1,o_2,o_3\}}}

\newcommand{\ch}{\textsc{ch}}

\newbox\ProofSym
\setbox\ProofSym=\hbox{%
	\unitlength=0.18ex%
	\begin{picture}(10,10)
	\put(0,0){\framebox(9,9){}}
	\put(0,3){\framebox(6,6){}}
	\end{picture}}

\title{Polygon Queries for Convex Hulls of Points\footnote{
This research was supported by 
the MSIT (Ministry of Science and ICT), Korea, under the SW Starlab support program (IITP-2017-0-00905) supervised by the IITP (Institute for Information \& communications Technology Promotion).}}

\author{Eunjin Oh\thanks{Max Planck Institute for Informatics, Saarbr\"ucken, Germany. \tt{eoh@mpi-inf.mpg.de}} \and Hee-Kap Ahn\thanks{Pohang University of Science and Technology, Pohang, Korea. \tt{heekap@postech.ac.kr}}}

\graphicspath{{figures/}}
\begin{document}
\date{}
	\maketitle

\begin{abstract}
  We study the following range searching problem: Preprocess a set $P$
  of $n$ points in the plane with respect to a set $\oset$ of $k$ orientations
  % , for a constant ,
  in the plane so that given an $\oset$-oriented convex polygon 
  $Q$, the convex hull of $P\cap Q$ can
  be computed efficiently, where an $\oset$-oriented polygon is a
  polygon whose edges have orientations in $\oset$.  We present a data
  structure with $O(nk^3\log^2n)$ space and $O(nk^3\log^2n)$ 
  construction time, and an $O(h+s\log^2 n)$-time query algorithm 
  for any query $\oset$-oriented convex $s$-gon $Q$, where $h$ is the complexity of the convex hull. 
  Also, we can compute the perimeter or area of the convex hull of $P\cap Q$ in $O(s\log^2n)$ time
  using the data structure. 
\end{abstract}

\section{Introduction}
Range searching is one of the most thoroughly
studied problems in computational geometry
for decades from 1970s. %Quadtrees, $kd$-trees, and range trees
Range trees and $kd$-trees 
were proposed as data structures for orthogonal range searching, 
and their sizes and query times have been improved over the years.
The most efficient data structures for orthogonal range searching
for points in the plane~\cite{Chazelle-1986} and  in higher
dimensions~\cite{Chazelle-1990} are due to
Chazelle.
% gave the most efficient data structure for
% range queries for $n$ points in the plane~\cite{Chazelle-1986}
% and in higher dimensions~\cite{Chazelle-1990}:
% data structures with $O(n\log n/\log\log n)$ space and $(\log n+s))$
% query time for 2-dimension and with $O(n(\log n/\log\log n)^{d-1})$ space
% and polylogarithmic query time for $d$-dimension,
% where $s$ is the number of reported points.

There are variants of the range searching problem that
allow other types of query ranges, such as circles or triangles.
Many of them can be solved using partition trees or a combination
of partition trees and cutting trees.
The simplex range searching problem, which is a higher dimensional
analogue of the triangular range searching, has gained much attention
in computational geometry as many other problems with more
general ranges can be reduced to it. %~\cite{Agarwal-1994}.
As an application, it can be used to solve the hidden surface removal in
computer graphics~\cite{Agarwal-1993,deBerg-1994}.
%\ccheck{A bit more on Simplex range searching and a special case - half-space range searching.}

The polygon range searching is a generalization of the simplex range searching
in which the search domain is a convex polygon.
Willard~\cite{Willard-1982} gave a data structure, called the polygon tree,
with $O(n)$ space and
an $O(n^{0.77})$-time algorithm for counting the number of
points lying inside an arbitrary query polygon of constant complexity. The query time 
was improved later by Edelsbrunner and Welzl~\cite{Edelsbrunner-1986} to $O(n^{0.695})$.
By using the stabbing numbers of spanning trees,
Chazelle and Welzl~\cite{Chazelle-1989} gave a data structure of size
$O(n\log n)$ with an \mbox{$O(\sqrt{kn}\log n)$-time} query algorithm for computing 
the number of points lying inside a query convex $k$-gon
for arbitrary values of $k$ with $k\leq n$. When $k$ is fixed 
for all queries, the size of the data structure drops to $O(n)$.
Quite a few heuristic techniques and frameworks have been proposed to
process polygon range queries on large-scale spatial data in a parallel and distributed
manner on top of MapReduce~\cite{Dean-2008}.
For overviews of results on range searching, see the survey 
by Agarwal and Erickson~\cite{Agarwal-1999}.

In this paper, we consider the following polygon range searching problem:
Preprocess a set $P$ of $n$ points with respect
to a set $\oset$ of $k$ orientations in the plane %, for a constant $c$, in the plane
so that given an $\oset$-oriented convex polygon $Q$,
the convex hull of $P\cap Q$, and its perimeter and area, can be
computed efficiently, where an $\oset$-oriented polygon is a polygon
whose edges have orientations in $\oset$.
Here, an \emph{orientation} is defined as a unit vector in the plane. We say
that an edge has an orientation $o$ if the edge is parallel to $o$.
%\complain{I added these two sentences because a reviewer said that we should define an orientation of an edge.}

Whereas orthogonal and simplex range queries can be carried out efficiently,
it is quite expensive for queries of arbitrary polygons in general. This is a
phenomenon that occurs in many other geometry problems. In an effort to
overcome such inefficiency and provide robust computation,
there have been quite a few works on ``finite orientation geometry'',
for instance, computing distances~\cite{Widmayer-1987} in fixed orientations,
finding the contour of the union of a collection of polygons with edges
of restricted orientations~\cite{Souvaine-1992}, and
constructing Voronoi diagrams~\cite{Agarwal-2015,Chen-2006} using
a distance metric induced by a convex $k$-gon. In the line of this research,
we suggests the polygon queries whose edges have orientations from a fixed set of
orientations. Such a polygon query, as an approximation of an arbitrary polygon,
can be used in appropriate areas of application, for instance, in VLSI-design, and
possibly takes advantages of the restricted number of orientations and robustness
in computation.

\paragraph{Previous Works.}
%For a set $P$ of $n$ points in the plane,
%Overmars and van Leeuwen~\cite{Overmars-1981}  gave a data structure
%which allows insertions and deletions of points in $P$ in $O(\log^2n)$ time and
%reporting the convex hull of $P$ in $O(\log n+h)$ time, where $h$ is the number of vertices
%of the hull. 
%
Brass et al.~\cite{Brass-2013} gave a data structure on $P$
for a query range $Q$ and a few geometric extent measures, including
the convex hull of $P\cap Q$ and its perimeter or area.
They gave a data structure with
$O(n\log^2 n)$ space and $O(n\log^3n)$ construction time
that given a query axis-parallel rectangle $Q$, 
reports the convex hull of $P\cap Q$ in $O(\log^5 n+h)$ time and
its perimeter or area in $O(\log^5n)$ time, where
$h$ is the complexity of the convex hull.

Both the data structure and query algorithm for reporting the convex hull of $P\cap Q$ 
were improved by Modiu et
al.~\cite{CCCG-Nadeem-2013}. They gave a data structure with
$O(n\log n)$ space and $O(n\log n)$ construction time that
given a query axis-parallel rectangle $Q$, reports the convex hull of
$P\cap Q$ in $O(\log^2 n+ h)$ time.

For computing the perimeter of the convex hull of $P\cap Q$,
the running time of the query algorithm by Brass et al. was improved by
Abrahamsen et al.~\cite{Abrahamsen-2017}.
For a query axis-parallel rectangle $Q$, their data structure supports $O(\log^3 n)$ query time.
Also, they presented a data structure of size $O(n\log^3 n)$ 
for supporting $O(\log^4 n)$ query time for a $5$-gon whose edges have three predetermined orientations. 

%The problem for reporting the convex hull of $P\cap Q$ for any triangle query $Q$
%was also studied~\cite{Oh-2017}. They provide a data structure of size $O(n)$ that supports
%$O(\sqrt{n}\log^2 n)$ query time for a triangle query.

% There have been results on geometry problems on polygons with edges
% of restricted orientations~\cite{Guting-1984}.  The problem of finding
% the contour of the union of a collection of polygons with edges of
% restricted orientations was considered by
% Souvaine~\cite{Souvaine-1992}.  Nielsen studied a point cover problem
% of such polygons~\cite{Nielsen-2001}.

\paragraph{Our Result.} Let $P$ be a set of $n$ points and let
$\oset$ be a set of $k$ orientations in the plane.
% for some constant
% $c$ distinct real values for some constant
% $c\in\mathbb{N}$,
% and  in the plane.
\begin{itemize}
\item We present a data structure on $P$ that allows us to compute the
  perimeter or area of the convex hull of points of $P$ contained in
  any query $\oset$-oriented convex $s$-gon in
    $O(s\log^2 n)$ time. We can construct the data structure with
  $O(nk^3\log^2n)$ space in $O(nk^3\log^2 n)$ time. Note that
  $s$ is at most $2k$ because $Q$ is convex.  % If $s$ is a
  When the query polygon has a constant complexity,
  as for the case of $\oset$-oriented triangle queries, the
    query time is only $O(\log^2 n)$.
\item For queries of reporting the convex hull of the points contained
  in a query $\oset$-oriented convex $s$-gon, the query algorithm
  takes $O(h)$ time in addition to the query times for the
    perimeter or area case, without increasing the size and
  construction time for the data structure, % remain the same,
  where $h$ is the complexity of the convex hull.
\item For $k=2$, % that is, $\oset$ consists of only two orientations,
  we can construct the data
    structure with $O(n\log n)$ space in $O(n\log n)$ time whose query
    time is $O(\log^2 n)$ for computing the perimeter or area of
    the convex hull of $P\cap Q$ and $O(\log^2 n+h)$ for reporting
    the convex hull.
\item Our data structure can be used to %as a subprocedure to
  improve the $O(n\log^4 n)$-time algorithm
  by Abrahamsen et al.~\cite{Abrahamsen-2017}
% They gave an algorithm using 
% to compute a
  for computing the minimum perimeter-sum bipartition of $P$.
  % a set $P$ of $n$ points in the plane.
  Their data structure requires $O(n\log^3 n)$ space and
  allows to compute the perimeter of the
  convex hull of points of $P$ contained in a $5$-gon whose edges
  have three predetermined orientations.
  % in \ccheck{a set of three distinct orientations.}
  If we replace their data structure % for computing the perimeter
  with ours, we can obtain an $O(n\log^2 n)$-time
algorithm for their problem using $O(n\log^2 n)$ space.
\end{itemize}

\section{Axis-Parallel Rectangle Queries for Convex Hulls}\label{sec:rect} 
We first consider axis-parallel rectangle queries. Given a set $P$ of $n$ points in
the plane, Modiu et al.~\cite{CCCG-Nadeem-2013} gave a data
structure on $P$ with $O(n\log n)$ space that reports the
convex hull of $P\cap Q$ in $O(\log^2 n+ h)$ time for any query axis-parallel 
rectangle $Q$, where $h$ is the complexity of the convex hull.  We
show that their data structure with a modification allows us to
compute the perimeter of the convex hull of $P\cap Q$ in
$O(\log^2 n)$ time. % We consider the perimeter case in this section. 

\subsection{Data Structure}\label{sec:DS}
We first briefly introduce the data structure given by Modiu et al.,
which is called a \emph{two-layer grid-like range tree}. 
To obtain a data structure for computing the parameter of the convex hull of $P\cap Q$
for a query axis-parallel rectangle $Q$,
we store information in each node of the two-layer grid-like range tree. 

\paragraph{Two-layer Grid-like Range Tree.}
The two-layer grid-like range tree is a variant of the two-layer standard range tree on $P$.
The two-layer standard range tree on $P$ is a two-level balanced
binary search tree~\cite{CGbook}.  The level-$1$ tree is a balanced
binary search tree $T_x$ on the points of $P$ with respect to their $x$-coordinates.
Each node $\alpha$ in $T_x$ corresponds to a vertical slab $I(\alpha)$. 
%Since $\alpha$ also corresponds to the slab
%$I(\alpha)\times (-\infty,+\infty)$ naturally, we abuse the notation $I(\alpha)$
%to denote the $x$-interval and the slab.
%\complain{Do we use $I(\alpha)$ to denote the $x$-interval?}
The node $\alpha$ has a balanced binary search tree on the points of 
$P\cap I(\alpha)$ with respect to their $y$-coordinates as its
level-2 tree.  In this way, each node $v$ in a level-2 tree
corresponds to an axis-parallel rectangle $B(v)$.

For any query axis-parallel rectangle $Q$, there is a set $\mathcal{V}$ of 
$O(\log^2 n)$ nodes of the level-2 trees such that the rectangles
$B(v)$ of $v\in\mathcal{V}$ are pairwise interior disjoint,
$Q\cap B(v)\neq\emptyset$ for every $v\in\mathcal{V}$, and
$\bigcup_{v\in \mathcal{V}}(P\cap B(v))=P\cap Q$.
For $v\in\mathcal{V}$,
we call $B(v)$ a \emph{canonical cell} for $Q$. 
One drawback of this structure is that the canonical cells for $Q$ are
not aligned with respect to their horizontal sides in general.
See Figure~\ref{fig:grid}(a).

 \begin{figure}
 	\begin{center}
 		\includegraphics[width=0.8\textwidth]{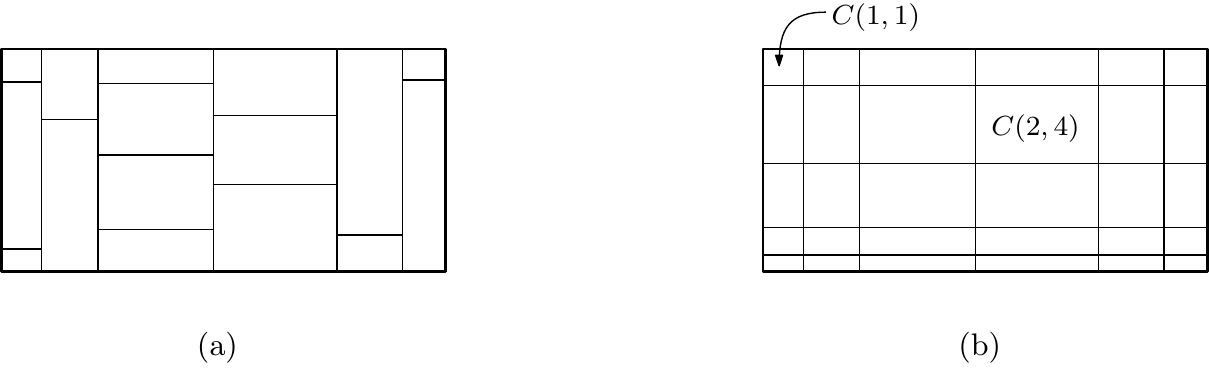}
 		\caption{\small\label{fig:grid} (a) Canonical cells in a standard range tree.
 		(b) Canonical cells in a grid-like range tree.}
 	\end{center}
 \end{figure}

To overcome this drawback, Modiu et al.~\cite{CCCG-Nadeem-2013} gave
the \emph{two-layer grid-like range tree} so that the canonical
cells for any query axis-parallel rectangle $Q$ are aligned
across all nodes $\alpha$ in the level-1 tree with
  $I(\alpha) \cap Q\neq\emptyset$. The two-layer grid-like
range tree is also a two-level tree whose level-1 tree is a balanced
binary search tree $T_x$ on the points of $P$ with respect to their $x$-coordinates.
%For a node $\alpha$ in $T_x$, let $P(\alpha)$ be
%the set of points of $P$ whose $x$-projections are in $I(\alpha)$. 
Each node $\alpha$ of $T_x$ is associated with the
level-2 tree $T_y(\alpha)$ which is a binary search tree
on the points of $P\cap I(\alpha)$. But, unlike the standard range
tree, $T_y(\alpha)$ is obtained from $T_y$ by removing the subtrees rooted at
all nodes whose corresponding rectangles have no point in $P\cap I(\alpha)$ and by contracting all nodes which have only
one child, where $T_y$ is a balanced binary search tree on the points of $P$
with respect to their $y$-coordinates.  Therefore, $T_y(\alpha)$ is not balanced but
a full binary tree of height $O(\log n)$, and
it is called a \emph{contracted tree} on $P\cap I(\alpha)$. 
By construction, the
canonical cells for any axis-parallel rectangle $Q$ are aligned.
% \ccheck{But in this case, the union of all canonical cells does not
%  necessarily coincide with $Q$.} \complain{???}  We simply treat the
% nodes of all level-2 trees of the range tree as the nodes of the range
% tree.

\begin{lemma}[\cite{CCCG-Nadeem-2013}]
  The two-layer grid-like range tree on a set of $n$ points in the plane
  can be computed in $O(n\log n)$ time. Moreover, its size is $O(n\log n)$.
\end{lemma}

\paragraph{Information Stored in a Node of a Level-2 Tree.}
To compute the perimeter of the convex hull of $P\cap Q$ for a query
axis-parallel rectangle $Q$ efficiently, we store additional information
for each node $v$ of the level-2 trees as follows. The node $v$ has two children in the
level-2 tree that $v$ belongs to. Let $u_1$ and $u_2$ be the two children of $v$ such that
$B(u_1)$ lies above $B(u_2)$.  By construction, $B(v)$ is partitioned
into $B(u_1)$ and $B(u_2)$.

Consider the convex hull $\ch(v)$ of $B(v)\cap P$ and the convex hull
$\ch(u_i)$ of $B(u_i)\cap P$ for $i=1,2$.  There are at most two edges
of $\ch(v)$ that appear on neither $\ch(u_1)$ nor $\ch(u_2)$.  We call
such an edge a \emph{bridge} of $\ch(v)$ with respect to $\ch(u_1)$
and $\ch(u_2)$, or simply a bridge of $\ch(v)$. Note that a bridge of
$\ch(v)$ has one endpoint on $\ch(u_1)$ and the other endpoint on
$\ch(u_2)$. We call the bridge of $\ch(v)$ whose clockwise endpoint
lies on $\ch(u_1)$ and counterclockwise endpoint lies on $\ch(u_2)$
along the boundary of $\ch(v)$ the \emph{$cw$-bridge} of $\ch(v)$. We call the
other bridge of $\ch(v)$ the \emph{$ccw$-bridge} of $\ch(v)$. 
See Figure~\ref{fig:basic-op}(a).

 \begin{figure}
   \begin{center}
     \includegraphics[width=0.7\textwidth]{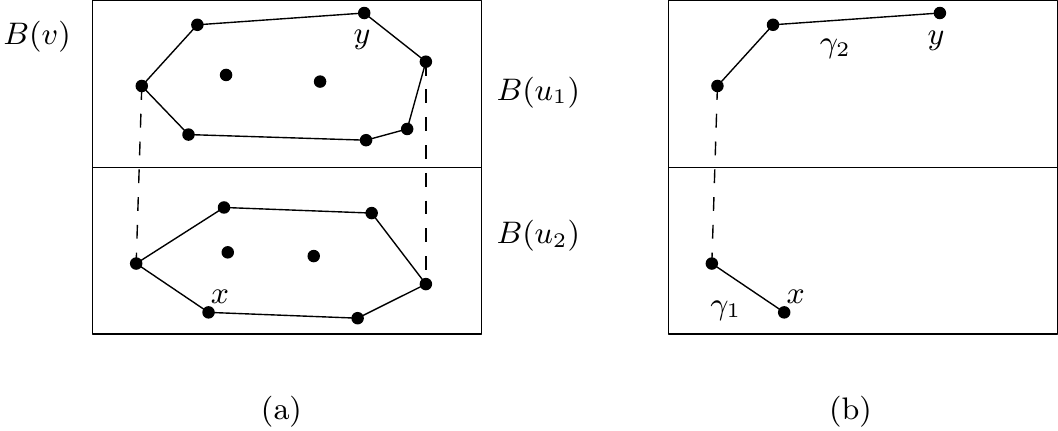}
     \caption{\small\label{fig:basic-op} (a) A node $v$ of the level-2
       tree has two children $u_1$ and $u_2$ such that
       $B(v)$ is partitioned into $B(u_1)$ $B(u_2)$. The dashed line segments
       are bridges stored in this node. The left one is the $cw$-bridge, and the right one
       is the $ccw$-bridge.
       (b) The part of $\ch(v)$ from $x$ to $y$ in clockwise order
       is decomposed into two polygonal curves $\gamma_1$ and $\gamma_2$
       with respect to the bridges
       of $\ch(v)$.}
   \end{center}
 \end{figure}

For each node $v$ of the level-2 trees, we store the two bridges of $\ch(v)$ and
the length of each polygonal chain of $\ch(v)$ lying between the two
bridges.  In addition, we store the length of each polygonal chain
connecting an endpoint $e$ of a bridge of $\ch(v)$ and an endpoint
$e'$ of a bridge of $\ch(p(v))$ for the parent node $p(v)$ of $v$ 
along the boundary of $\ch(v)$ if $e$ and $e'$ are contained in $B(v)$.  We do this for every
pair of the endpoints of the bridges of $\ch(v)$ and $\ch(p(v))$ that
are contained in $B(v)$.  Since only a constant number of bridges are
involved, the information stored for $v$ is also of constant size. Each
bridge can be computed in time linear in the number of vertices of
$\ch(u)$ which do not appear on $\ch(v)$ for a child $u$ of $v$.  The
length of each polygonal chain we store for $v$ can also be computed
in this time.  Notice that a vertex of $\ch(u)$ which does not appear on
$\ch(v)$ does not appear on $\ch(v')$ for any ancestor $v'$ of
$v$. Therefore, the total running time for computing the bridges is
linear in the total number of points corresponding to the leaf nodes of the level-2 trees, which is $O(n\log^2 n)$.
 
%\complain{In constant time? Or $O(\log n)$ time?}
% \ccheck{This information gives a binary search tree of the edges of $\ch(v)$ of height $O(\log n)$
%   for each node $v$ of the range tree.} \complain{This sentence reads strange.
% It must be as follows. ``Using the information, we can compute the length of ...'' I suggest to delete it as this is described in the following lemma.}
% Therefore, we have the following operation.

%We call the bridge of $\ch(v)$ whose clockwise endpoint
%lies on $\ch(u_1)$ and counterclockwise endpoint lies on $\ch(u_2)$
%along the boundary of $\ch(v)$ the \emph{$cw$-bridge} of $\ch(v)$. We call the
%other bridge of $\ch(v)$ the \emph{$ccw$-bridge} of $\ch(v)$.

We will use the following lemma for our query algorithm. 

\begin{lemma}\label{lem:basic-op}
  Given a node $v$ of a level-2 tree and two vertices $x, y$ of
  $\ch(v)$, we can compute the length of the part of the boundary of
  $\ch(v)$ from $x$ to $y$ in clockwise order along the boundary of
  $\ch(v)$ in $O(\log n)$ time.
\end{lemma}
\begin{proof}
	Let $\gamma$ be the part of the boundary of $\ch(v)$ from $x$ to $y$
	in clockwise order along the boundary of $\ch(v)$.  
	Let $u_1$ and $u_2$ be the children of $v$ such that $B(u_1)$ lies
	above $B(u_2)$.
	We consider the case that only one bridge lies on $\gamma$.  
	We assume further that $x$ is contained in $B(u_2)$, and $y$ is contained in
	$B(u_1)$. The other cases can be handled analogously.
	%Since $\gamma$ contains the $cw$-bridge of $\ch(v)$ only, 
	See Figure~\ref{fig:basic-op}(b).
	Then $\gamma$, excluding
	the bridge $b$ of $\ch(v)$ lying on $\gamma$, consists of two polygonal curves,
	$\gamma_1$ and $\gamma_2$, with $\gamma_1$ lying before $b$ and
	$\gamma_2$ after $b$ along $\gamma$ from $x$.
	
	We show how to compute the length of $\gamma_1$ only. The same method works for
	$\gamma_2$.  To do this, we traverse the
	level-2 tree along the path from the root to the leaf corresponding to $x$ and process nodes as
	follows. For each node $v'$ on the path, 
	our task is computing the length of a polygonal chain of $\ch(v')$ connecting 
	$x$ and an endpoint of a bridge of $\ch(v')$. 
	
	We first consider the case that $\gamma_1$ contains a bridge
	of $\ch(u_2)$.  The chain $\gamma_1$, excluding the bridges of
	$\ch(u_2)$, consists of at most three pieces because there are
	at most two bridges of $\ch(u_2)$. One of the pieces has
	one endpoint on $x$ and the other on an endpoint of a
	bridge of $\ch(u_2)$, and each of the other pieces has 
	one endpoint on an endpoint of a bridge of $\ch(u_2)$ and
	the other on an endpoint of a bridge of %$\ch(u_2)$ and
	$\ch(v)$. Therefore, the lengths of the pieces of
	$\gamma_1$ which are not incident to $x$ are stored in $u_2$.
	Thus, it suffices to compute the piece of $\gamma_1$ with
	endpoints on $x$ and an endpoint of a bridge of
	$\ch(u_2)$. To do this, we visit the child $w$ of $u_2$ such
	that $B(w)$ contains $x$ and 
	compute the length of the piece of $\gamma_1$ recursively.
	
	Consider the case that $\gamma_1$ contains no bridge of $\ch(u_2)$.
	In this case, we find the first endpoint $e$ of a bridge of $\ch(u_2)$
	that appears first along its boundary from $x$ in clockwise order.
	The length of $\gamma_1$ is equal
	to the length of the part $\gamma_{xe}$ of the boundary of $\ch(u_2)$
	from $x$ to $e$ in clockwise order minus the length of the part $\gamma_{ze}$
	of the boundary
	from the endpoint $z$ of $\gamma_1$ other than $x$ to $e$ in clockwise order.
	The length of $\gamma_{ze}$ is stored in $u_2$. Thus it suffices to compute the length of 
	$\gamma_{xe}$. Since $\gamma_{xe}$ connects $x$ and an endpoint, $e$,
	of a bridge of $\ch(u_2)$, we visit the child $w$ of $u_2$ such that
	$B(w)$ contains $x$ and compute the length of $\gamma_{xe}$ recursively.
	
	In this way, we traverse the tree
	along the path from $v$ to a leaf node in $O(\log n)$ time. Finally, we
	obtain the length of $\gamma_1$.
\end{proof}

\subsection{Query Algorithm}\label{sec:query-2d}
Let $Q$ be an axis-parallel rectangle. We present an algorithm for
computing the perimeter of the convex hull of $P\cap Q$ in
$O(\log^2 n)$ time.  We call the part of the convex hull from its
topmost vertex to its rightmost vertex in clockwise order along its
boundary the \emph{urc-hull} of $P\cap Q$.  In the following, we
compute the length of the urc-hull $\gamma$ of $P\cap Q$ in
$O(\log^2 n)$ time.  The lengths of the other parts of the convex hull
of $P\cap Q$ can be computed analogously.

We use the algorithm by Overmars and van Leeuwen~\cite{Overmars-1981}
for computing the outer tangents between any two convex polygons.
% in
% $O(\log n)$ time, where $n$ is the total complexity of the convex
% polygons.
\begin{lemma}[\cite{Overmars-1981}]\label{lem:overmars}
  Given any two convex polygons stored in two binary search trees of height $O(\log n)$,
  we can compute the outer tangents between them in $O(\log n)$ time,
  where $n$ is the total complexity of the convex hulls.
\end{lemma}

We compute the set $\mathcal{V}$ of the canonical cells for $Q$ in
$O(\log^2 n)$ time.
% and denote the set of them by $\mathcal{V}$.
Recall that the size of $\mathcal{V}$ is $O(\log^2 n)$.  We consider
the cells of $\mathcal{V}$ as grid cells of a grid with $O(\log n)$
rows and $O(\log n)$ columns.  We use $C(i,j)$ to denote the grid cell
of the $i$th row and $j$th column such that the leftmost cell in the
topmost row is $C(1,1)$. See Figure~\ref{fig:grid}(b). 
Notice that a grid cell $C(i,j)$ might not be contained in
$\mathcal{V}$.

Recall that we want to compute the urc-hull of $P\cap Q$.  To do this,
we compute the point $p_x$ with largest $x$-coordinate and the point
$p_y$ with largest $y$-coordinate from $P\cap Q$ in $O(\log n)$ time
using the range tree~\cite{CGbook}. Then we find the cells of
$\mathcal{V}$ containing each of them in the same time. Let
$C(i_1,j_1)$ and $C(i_2,j_2)$ be the cells of $\mathcal{V}$ containing
$p_y$ and $p_x$, respectively.

We traverse the cells of $\mathcal{V}$ starting from $C(i_1,j_1)$
until we reach $C(i_2,j_2)$ as follows.  We find every cell
$C(i,j)\in\mathcal{V}$ with $i_1\leq i\leq i_2$ and
$j_1\leq j\leq j_2$ such that no cell
$C(i',j')$ with $i<i'$  is in $\mathcal{V}$
or no cell
$C(i',j')$ with $j>j'$ is in $\mathcal{V}$.
There are $O(\log n)$ such cells, and we call them \emph{extreme cells}. 
We can compute all extreme cells in $O(\log^2 n)$ time. 
Note that the urc-hull of $P\cap Q$ 
is the urc-hull of points contained in the extreme cells. 
To compute the urc-hull of $P\cap Q$, we traverse the extreme cells
in the lexicographical order with respect to the first index and then
the second index. %Note that we visit $O(\log n)$ cells in total.

%\begin{lemma}
%	If $C(i,j)\in\mathcal{V}$ contains a vertex of the urc-hull, it is an
%	extreme cell.
%\end{lemma}
%\begin{proof}
%\end{proof}

During the traversal, we maintain the urc-hull of the points contained
in the cells we visited so far using a binary search tree of height
$O(\log n)$. Imagine that we have just visited a cell
$C\in\mathcal{V}$ in the traversal.  Let $\delta_1$ denote the urc-hulls 
of the points contained in the cells we visited before the visit to $C$.
Let $\delta_2$ denote the urc-hulls of the points contained in the cells we visited
so far, including $C$.
Due to the data structure we maintained, we have a binary search tree of height $O(\log n)$
for the convex hull $\ch$ of the points contained in $C$. Moreover, we
have a binary search tree of height $O(\log n)$ for $\delta_1$ from
the traversal to the cells we visited so far.  Therefore, we compute
the outer tangents (bridges) between them in $O(\log n)$ time by
Lemma~\ref{lem:overmars}.  The urc-hull $\delta_2$ is the
concatenation of three polygonal curves: a part of $\ch$, the bridge,
and a part of $\delta_1$. Thus we can represent $\delta_2$ using a
binary search tree of height one plus the maximum of the
heights of the binary search trees for $\ch$ and $\delta_1$.

Since we traverse $O(\log n)$ cells in total, we obtain a binary
search tree of height $O(\log n)$ representing the urc-hull of
$P\cap Q$ after the traversal. The traversal takes $O(\log^2 n)$ time.
Notice that the urc-hull consists of $O(\log n)$ polygonal curves that
are parts from the convex hulls stored in cells of $\mathcal{V}$ and
$O(\log n)$ bridges connecting them. We can compute the length of the
polygonal curve in $O(\log^2 n)$ time in total by
Lemma~\ref{lem:basic-op}.

% Therefore, we have the following theorem.
\begin{theorem}
  Given a set $P$ of $n$ points in the plane, we can construct a data
  structure with $O(n\log n)$ space in $O(n\log n)$-time 
  preprocessing that allows us to compute the perimeter of the convex
  hull of $P\cap Q$ in $O(\log^2 n)$ time 
  for any query axis-parallel rectangle $Q$. 
\end{theorem}

Since the data structure with its construction and the query
algorithm can be used for any pair of orientations which are not
necessarily orthogonal through an affine transformation, they work
for any pair of orientations with the same space and time
complexities.

\begin{corollary}\label{cor:two-orientations}
  Given a set $P$ of $n$ points and a set $\oset$ of two orientations
  in the plane, we can construct a data
  structure with $O(n\log n)$ space in $O(n\log n)$-time 
  preprocessing that allows us to compute the perimeter of the convex
  hull of $P\cap Q$ in $O(\log^2 n)$ time 
  for any query $\oset$-oriented rectangle $Q$.
\end{corollary}
\section{\texorpdfstring{$\oset$}{O}-oriented Triangle Queries for Convex
  Hulls}\label{sec:tri}
In this section, we are given a set $P$ of $n$ points and a set
$\oset$ of $k$ distinct orientations  in the plane.
We preprocess the two sets so that we
can compute the perimeter of $P\cap Q$ for any query $\oset$-oriented
triangle $Q$ in the
plane efficiently. % whose edges have their orientations in $\oset$.
% We call a polygon whose edges have their orientations in $\oset$ an
% \emph{$\oset$-polygon}.
We construct a \emph{three-layer grid-like range tree} on $P$ with
respect to every 3-tuple $(o_1,o_2,o_3)$ of the orientations in
$\oset$, which is a generalization of the two-layer grid-like range
tree described in Section~\ref{sec:DS}.
% We construct a data structure for every 3-tuple $(o_1,o_2,o_3)$ of
% the orientations in $\oset$.  This data structure is a three-level
% tree such that a level-$i$ tree is constructed with respect to $o_i$
% for $i=1,2,3$.
A straightforward query algorithm takes $O(\log^3 n)$ time since there are $O(\log^2 n)$ canonical cells for a
query $\{o_1,o_2,o_3\}$-oriented triangle $Q$. 
However, it is unclear how to obtain a faster query algorithm as
 the query algorithm described in Section~\ref{sec:rect} does
not generalize to this problem directly.  A main reason is that a
canonical cell for any query $\{o_1,o_2,o_3\}$-oriented triangle is a
$\{o_1,o_2,o_3\}$-oriented polygon, not a parallelogram.
% In other words, such a canonical cell is defined with respect to
% three axes.
This makes it unclear how to apply the approach in Section~\ref{sec:rect} to this
case.

In this section, we present an $O(\log^2n)$-time query algorithm for this problem.
Our algorithm improves this straightforward algorithm by a factor of $\log n$.
To do this, we classify canonical cells for $Q$ into two types.
We can handle the cells of the first type as we do in Section~\ref{sec:rect}
and compute the convex hull of the points of $P$ contained in them.
Then we handle the cells of the second type by defining a specific ordering to these cells
so that we can compute the convex hull of the points of $P$ contained
in them efficiently.
% apply a procedure similar to Graham's scan approach.
Then we merge the two convex hulls to obtain the convex hull of $P\cap Q$.

\subsection{Data Structure}
We construct a \emph{three-layer grid-like range tree} on $P$
with respect to every 3-tuple of the orientations in $\oset$.  Let
$(o_1,o_2,o_3)$ be a 3-tuple of the orientations in $\oset$.  For an
index $i=1,2,3$, we call the projection of a point in the plane 
onto a line orthogonal to $o_i$ the \emph{$o_i$-projection} of the
point.  Let $T_i$ be a balanced binary search tree on the
$o_i$-projections of the points of $P$ for $i=1,2,3$.

\paragraph{Three-layer Grid-like Range Tree.}
The level-1 tree of the grid-like range tree is $T_1$. Each node of
$T_1$ corresponds to a slab of orientation $o_1$. For each node of the level-1 tree, we construct
a contracted tree of the $o_2$-projections of the points contained
in the slab. % corresponding to the node with respect to $T_2$. 
A node of a level-$2$ tree corresponds to an $\{o_1,o_2\}$-oriented parallelogram. For
each node of a level-2 tree, we construct a contracted tree of the
$o_3$-projections of the points contained in the $\{o_1,o_2\}$-oriented
parallelogram. % corresponding to the node with respect to $T_3$.  
A node $v$ of a level-3 tree corresponds to an $\otuple$-oriented polygon 
$B(v)$ with at most six vertices.
See Figure~\ref{fig:partition} for an illustration.
% We simply call $B(v)$ a cell of the range tree for every node
%of a level-3 tree.
\begin{figure}
  \begin{center}
    \includegraphics[width=.8\textwidth]{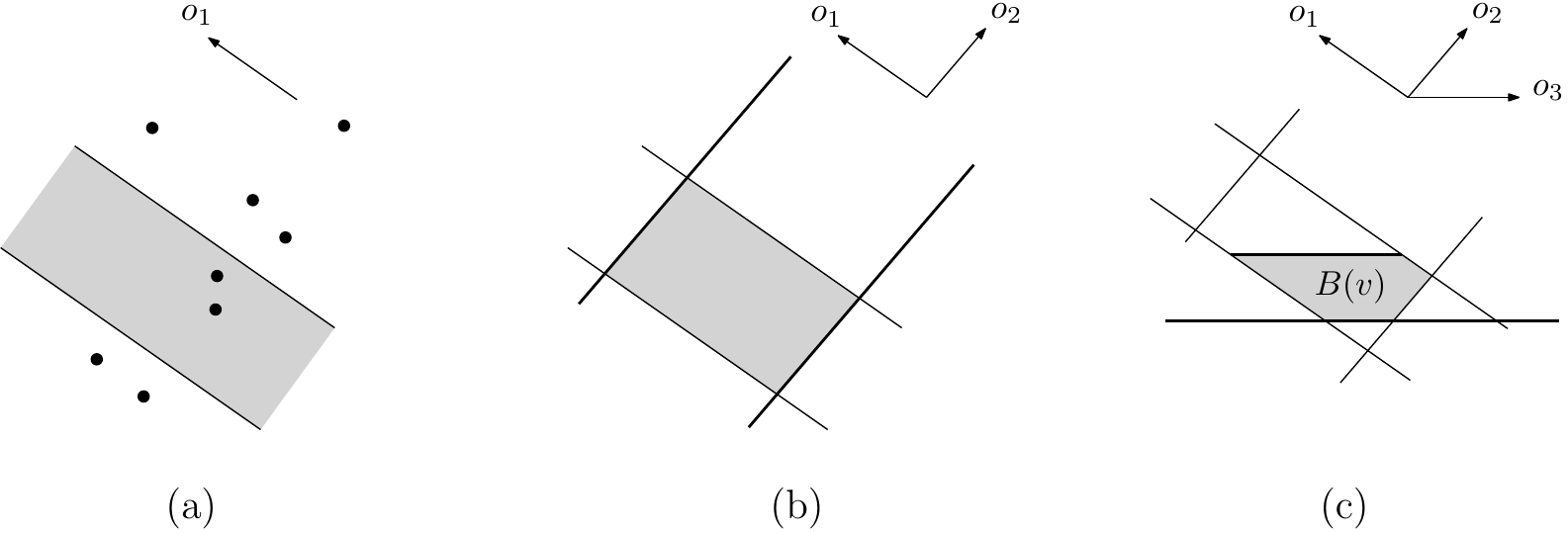}
    \caption{\small\label{fig:partition} (a) A node
    	of $T_1$ corresponds to a slab of orientation $o_1$. 
    	(b) A node of a level-2 tree corresponds to a parallelogram having  
    	two sides of orientation $o_1$ and two sides of orientation $o_2$. 
    	(c) A node $v$ of a level-3 tree corresponds to an $\{o_1,o_2,o_3\}$-polygon $B(v)$.}
  \end{center}
\end{figure}

\paragraph{Information Stored in a Node of a Level-3 Tree.}
Without loss of generality, we assume $o_3$ is parallel to the $x$-axis.
To compute the perimeter of the convex hull of $P\cap Q$ for a query $\oset$-oriented
triangle $Q$, we store additional information for each node $v$ of a \mbox{level-3} 
tree as follows.
The node $v$ has two children $u_1$ and $u_2$ in the level-3 tree
that $v$ belongs to such that $B(u_1)$ lies above $B(u_2)$.
By construction, $B(v)$ is partitioned into $B(u_1)$ and $B(u_2)$.
See Figure~\ref{fig:information} for an illustration.
\begin{figure}
  \begin{center}
    \includegraphics[width=.6\textwidth]{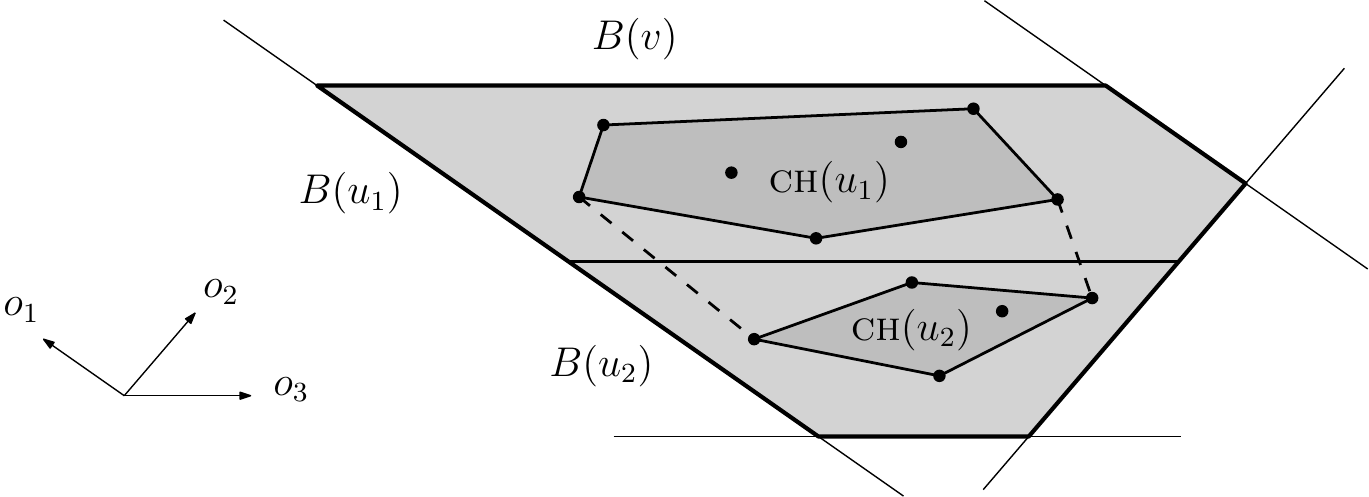}
    \caption{\small\label{fig:information} %Information stored on each node $v$ of $T_3(\beta)$.
A node $v$ of a level-3 tree has two children $u_1$ and $u_2$ such that $B(v)$ is partitioned into $B(u_1)$ and $B(u_2)$
with $B(u_1)$ lying above $B(u_2)$. There are at most two bridges of $\textsc{ch}(v)$, each has one endpoint on $\textsc{ch}(u_1)$ and the other on $\textsc{ch}(u_2)$.}
  \end{center}
\end{figure}

Consider the convex hull $\ch(v)$ of $P\cap B(v)$ and the convex hull
$\ch(u_i)$ of $P\cap B(u_i)$ for $i=1,2$.  There are at most two edges
of $\ch(v)$ that appear on neither $\ch(u_1)$ nor $\ch(u_2)$.  We call
such an edge a \emph{bridge} of $\ch(v)$ with respect to $\ch(u_1)$
and $\ch(u_2)$, or simply a bridge of $\ch(v)$.
Note that a bridge of
$\ch(v)$ has one endpoint on $\ch(u_1)$ and the other endpoint on
$\ch(u_2)$.
% We call the bridge of $\ch(v)$ whose clockwise endpoint
% lies on $\ch(u_1)$ and counterclockwise endpoint lies on $\ch(u_2)$
% along the boundary of $\ch(v)$ the \emph{$cw$-bridge}. We call the
% other bridge of $\ch(v)$ the \emph{$ccw$-bridge}.
As we do in Section~\ref{sec:rect}, for each node $v$ of the level-3
trees, we store two bridges of $\ch(v)$ and the length of each polygonal chain
of $\ch(v)$ lying between the two bridges.  Also, we store the length of
each polygonal chain connecting an endpoint of a bridge of $\ch(v)$
and an endpoint of a bridge of $\ch(p(v))$ for the parent $p(v)$ of
$v$ along the boundary of $\ch(v)$ if the two endpoints appear on $\ch(v)$.
%Due to this information, we can obtain
%a binary search tree of the edges of $\ch(v)$ of height $O(\log n)$ for any
%node $v$ of the level-3 tree.
The following lemma can be proven in a way similar to Lemma~\ref{lem:basic-op}.

\begin{lemma}\label{lem:basic-op-3d}
	Given a node $v$ of a level-3 tree and two vertices $x, y$ of
	$\ch(v)$, we can compute the length of the part of the boundary of
	$\ch(v)$ from $x$ to $y$ in clockwise order along the boundary of
	$\ch(v)$ in $O(\log n)$ time.
\end{lemma}

\subsection{Query Algorithm} 
In this subsection, we present an $O(\log^2 n)$-time query algorithm
for computing the perimeter of the convex hull of $P\cap Q$ for a query $\otuple$-oriented
triangle $Q$. 
Let $T$ be the three-layer grid-like range tree constructed
with respect to $(o_1,o_2,o_3)$.

\subsubsection{Computing Canonical Cells}
We obtain $O(\log^2 n)$ cells of $T$, called \emph{canonical cells} of $Q$, such that the union of
$P\cap C$ coincides with $P\cap Q$ for all the canonical cells $C$ as follows.
We first search the level-1 tree of $T$ along the endpoints of the $o_1$-projection of $Q$.
Then we obtain $O(\log n)$ nodes such that the union of the slabs corresponding to the nodes
contains $Q$. Then we search the level-2 tree associated with each such node along
the endpoints of the $o_2$-projection of $Q$.
Then we obtain $O(\log^2 n)$ nodes in total such that the union of the $\{o_1,o_2\}$-parallelograms corresponding
to the nodes contains $Q$. We discard all $\{o_1,o_2\}$-parallelograms not intersecting $Q$. 
Some of the remaining $\{o_1,o_2\}$-parallelograms are contained in $Q$, but the others
intersect the boundary of $Q$ in their interiors. For the nodes corresponding to 
the $\{o_1,o_2\}$-parallelograms intersecting the boundary of $Q$, we search their
level-3 trees along the $o_3$-projection of $Q$. 

As a result, we obtain 
$\{o_1,o_2\}$-parallelograms from the level-2 trees and $\otuple$-polygons from the level-3 trees
of size $O(\log^2 n)$ in total. See Figure~\ref{fig:canonical-cell-3d}. 
We call them the \emph{canonical cells} of $Q$ and denote the set
of them by $\mathcal{V}$. 
Also, we use $\mathcal{V}_p$ and $\mathcal{V}_h$ to denote the subsets of $\mathcal{V}$ consisting of
$\{o_1,o_2\}$-parallelograms from the level-2 trees and $\otuple$-polygons from the level-3 trees, respectively.
We can compute them in $O(\log^2 n)$ time.

\begin{figure}
  \begin{center}
    \includegraphics[width=.8\textwidth]{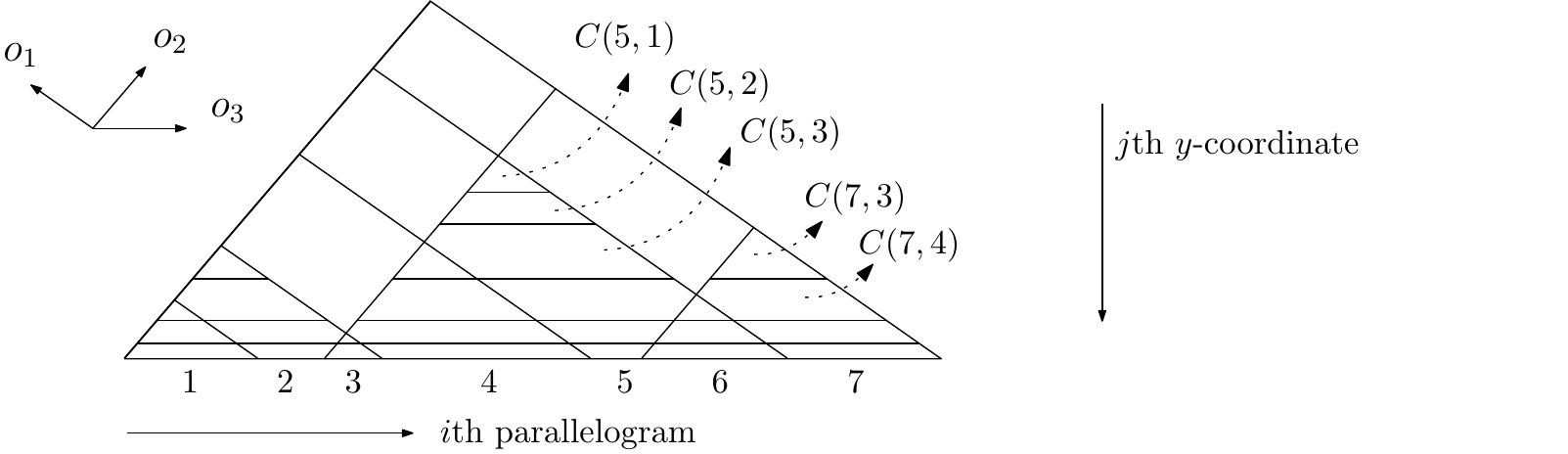}
    \caption{\small\label{fig:canonical-cell-3d} Canonical cells for a
      triangle.  Four $\{o_1,o_2\}$-oriented parallelogram cells from
      level-2 trees and 26 $\{o_1,o_2,o_3\}$-oriented
        polygon cells from level-3 trees.}
  \end{center}
\end{figure}

% \ccheck{\bigskip
% In Section~\ref{sec:query-2d}, we traverse the cells in a specific order and
% apply an algorithm similar to \ccheck{Graham's scan algorithm} for computing
% the convex hull of  $P\cap Q$. \complain{Graham's scan?}
% In this case, it is unclear how to use \ccheck{Graham's approach} to the
% cells of $\mathcal{V}$.  This is because we have three axes each of
% which is parallel to $o_i$ for $i=1,2,3$.
% Instead, we partition $\mathcal{V}$ into two subsets $\mathcal{V}_p$
% and $\mathcal{V}_h$ and compute the convex hull of the points of $P$
% contained in each subset. Then we merge the two convex hulls to obtain
% the convex hull of $P\cap Q$.} \complain{I don't think we need this. This is already
% stated in the beginning of Section 3.}

\subsubsection{Computing Convex Hulls for Each Subset}
We first compute the convex hull $\ch_p$ of the points contained in the cells of
 $\mathcal{V}_p$ and the convex hull $\ch_h$ of the points contained in the
 cells of $\mathcal{V}_h$. Then we merge them into the convex hull of $P\cap Q$ in Section~\ref{sec:merge}.
We can compute $\ch_p$ %of the points contained in the cells of $\mathcal{V}_p$
in $O(\log^2 n)$ time due to Corollary~\ref{cor:two-orientations}. This is because
the cells are aligned with respect to two axes which are parallel to $o_1$ and $o_2$ each.
Then we obtain a binary search tree of height $O(\log n)$ representing $\ch_p$.
Thus in the following, we focus on compute $\ch_h$.

Without loss of generality, assume that $Q$ lies above the
$x$-axis. Let $\ell$ be the side of $Q$ of orientation $o_3$.  We
assign a pair of indices to each cell of $\mathcal{V}_h$, which
consists of a \emph{row index} and a \emph{column index} as follows.
The cells of $\mathcal{V}_h$ come from $O(\log n)$ level-3 trees 
of the range tree. This means that each cell of $\mathcal{V}_h$ is contained
in the cell corresponding to the root of one of such level-3 trees.
These root cells are pairwise interior disjoint and intersect $\ell$.
% Every parallelogram cell intersects $\ell$, and they are pairwise interior disjoint.
% We assign an index $i$ to the $i$th leftmost root cell along $\ell$.
For each cell $v$ of $\mathcal{V}_h$ contained in the $i$th leftmost root cell along $\ell$,
we assign $i$ to it as the row index of $v$.
% See Figure~\ref{fig:canonical-cell-3d}.
% A cell of $\mathcal{V}_h$ is a trapezoid whose
% Each trapezoid cell has parallel sides are parallel to the $x$-axis.
The bottom side of a cell of $\mathcal{V}_h$ 
is parallel to the $x$-axis. Consider the $y$-coordinates of all bottom sides of the cells
of $\mathcal{V}_h$.  By construction, there are $O(\log n)$ distinct
$y$-coordinates although the size of $\mathcal{V}_h$ is $O(\log^2 n)$.
We assign an index $j$ to the cells of $\mathcal{V}_h$ whose bottom side
has the $j$th largest $y$-coordinates as their column indices.
Then each cell of $\mathcal{V}_h$ has an index $(i,j)$, where $i$ is its
row index and $j$ is its column index. Any two distinct cells of
$\mathcal{V}_h$ have distinct indices.  We let $C(i,j)$ be the cell of
$\mathcal{V}_h$ with index $(i,j)$.
%See Figure~\ref{fig:canonical-cell-3d}.

Due to the indices we assigned, we can apply a procedure similar to
Graham's scan algorithm for computing $\ch_h$.  We show how to compute
the urc-hull of $\ch_h$ only. The other parts of the boundary of
$\ch_h$ can be computed analogously.
%To do this, we compute the point $p_x$
%with largest $x$-coordinate and the point $p_y$ with largest $y$-coordinate
%among  the points contained in the cells of $\mathcal{V}$ in $O(\log^2 n)$ time
%by storing the topmost and rightmost points contained in each cell of the grid-like range tree
%as a preprocessing.
To do this, we choose $O(\log n)$ cells as follows.
Note that a cell of $\mathcal{V}_h$ is a polygon with at most $6$ vertices.
A trapezoid cell
$C(i,j)$ of $\mathcal{V}_h$ is called an \emph{extreme cell}
if there is no cell
$C(i',j')\in\mathcal{V}_h$ such that $i<i'$ and $j>j'$, or
% if there is no cell $C(i',j')\in\mathcal{V}_h$ such that
$i<i'$ and $j<j'$. Here, we need the disjunction. Otherwise, we cannot find 
some trapezoidal cell containing a vertex of the urc-hull. See Figure~\ref{fig:canonical-cell-ch}.
There are $O(\log n)$ extreme cells of $\mathcal{V}_h$.
In addition to these extreme cells, we choose every cell of $\mathcal{V}_h$ which are
not trapezoids, that is, convex $t$-gons with $t=3, 5, 6$. 
Note that there are $O(\log n)$ such cells because
such cells are incident to the corners of the cells of $\mathcal{V}_p$.
%\complain{Are these $t$-gons extreme cells? Why do we choose them all? How many such $t$-gons do we select?
%  $O(\log n)$? Why not more than that?
%  Or do you mean ``We choose all $t$-gons
%adjacent to the extreme trapezoid cells (or to the parallelogram cells in $\mathcal{V}_p$).''?}
In this way, we choose $O(\log n)$ cells of $\mathcal{V}_h$ in total.
\begin{figure}
  \begin{center}
    \includegraphics[width=.8\textwidth]{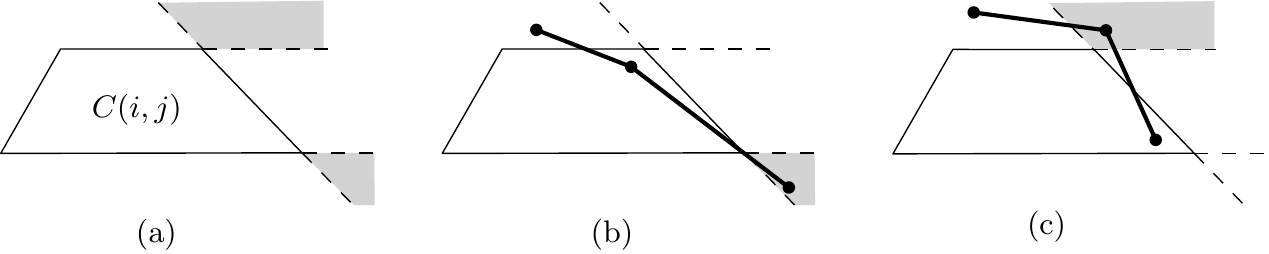}
    \caption{\small\label{fig:canonical-cell-ch}
    	(a) We choose $C(i,j)$ if and only if at least one of the two gray regions
    	contains no cell of $\mathcal{V}_h$.     	
    	(b) It is not sufficient to choose
        the only cells such that no cell of $\mathcal{V}_h$ is 
        contained in their lower gray regions 
        since the urc-hull might have its vertices in a cell whose lower gray regions
        contains a cell of $\mathcal{V}_h$.
        (c) Similarly, it is not sufficient to choose
        the only cells such that no cell of $\mathcal{V}_h$ is contained in
        their upper gray regions.}
  \end{center}
\end{figure}

\begin{lemma}\label{lem:correctness}
	A cell of $\mathcal{V}_h$ containing a vertex of the urc-hull of $\ch_h$ is 
	an extreme cell of $\mathcal{V}_h$ if it is a trapezoid.
\end{lemma}
\begin{proof}
	Let $v$ be a vertex of the urc-hull of $\ch_h$ and $C=C(i,j)$ be the trapezoid cell
	of $\mathcal{V}_h$ containing $v$.
	%  Assume to the contrary that $C$ is not an extreme cell of $\mathcal{V}$ but
	%  $C$ is a \ccheck{trapezoid.}
	Consider the region $H$ lying to the right of the line containing the right side of $C$.
	The lines containing the top and bottom sides of $C$ subdivide $H$ into three subregions.
	%\ccheck{By construction, the topmost and bottommost subregions contain a point of $P\cap Q$.}
	Since $v$ is a vertex of the urc-hull, the topmost or bottommost subregion
	contains no point of $P\cap Q$, that is, there is no cell $C(i',j')\in\mathcal{V}_h$
	such that $i<i'$ and $j>j'$, or %there is no cell $C(i',j')\in\mathcal{V}_h$ with 
	$i<i'$ and $j<j'$. Therefore, $C$ is an extreme cell.
	%   Therefore, $C$ does not contain any vertex of the urc-hull, which is a contradiction.
\end{proof}

By Lemma~\ref{lem:correctness}, the convex hull $\ch_h$ coincides with the convex hull of the convex hulls
of points in the cells chosen by the previous procedure. 
For each column $j$, we consider the cells with column index $j$ chosen by the previous procedure
one by one in increasing order with respect to their row indices, and compute the convex hull of points
contained in those cells. Then we consider the column indices one by one in increasing order, and compute
the convex hull of the convex hulls for column indices.
This takes $O(\log^2 n)$ time in total as we do in Section~\ref{sec:query-2d}.

In this way, we can obtain
a binary search tree of height $O(\log n)$ representing the urc-hull
of $\ch_h$.  The urc-hull consists of $O(\log n)$ polygonal
curves that are parts of the boundaries of the convex hulls stored in cells of
$\mathcal{V}_h$ and $O(\log n)$ bridges connecting them.  Therefore,
we can compute the lengths of the polygonal curves in $O(\log^2 n)$ time in total.  

% \paragraph{Remark.} 
% \blue{
% We can define the canonical cells for $Q$ in a different way.
% Specifically, we can search further the level-3 trees associated with the cells of $\mathcal{V}_p$.
% Then every cell becomes a $\{o_1,o_2,o_3\}$-parallelogon aligned with three axes. 
% The number of the canonical cells is still $O(\log^2 n)$. 
% One might wonder whether we can similarly define $O(\log n)$ cells which are \emph{extreme} in this case.
% \ccheck{It is possible to define such cells and compute the urc-hull in the same time
% without classifying canonical cells into two types.
% However, we need a more complicated procedure to compute such cells than the procedure described in this paper.} \complain{Does this paragraph help?
% Remove this paragraph ``Remark.''?}
% }

\subsubsection{Merging the Two Convex Hulls}\label{sec:merge}
The convex hull $\ch$ of $P\cap Q$ coincides with the convex hull of $\ch_p$
and $\ch_h$. To compute it, we need the following lemma.

\begin{lemma}\label{lem:intersection}
  The boundary of $\ch_p$ intersects the boundary of $\ch_h$ at most
  $O(\log n)$ times. We can compute the intersection points in $O(\log^2 n)$ time in total.
\end{lemma}
\begin{proof}
	Consider two edges, one from $\ch_p$ and one from $\ch_h$, intersecting each other.
	One of them is a bridge with endpoints lying on two distinct cells of $\mathcal{V}$.
	This is because the cells of $\mathcal{V}$ are pairwise interior disjoint.
	Moreover, each bridge in $\ch_p$ (or $\ch_h$) intersects the boundary of $\ch_h$ (or $\ch_p$) at most twice since $\ch_h$ and $\ch_p$ are convex. 
	Since there are $O(\log n)$ bridges in $\ch_p$ and $\ch_h$,
	there are $O(\log n)$ intersection points between the boundary of $\ch_p$ and the boundary of $\ch_h$.
	
	To compute the intersection points, we compute the intersection points between each bridge of a convex hull  
	and the boundary of the other convex hull. For each bridge, we can compute the two 
	intersection points in $O(\log n)$ time since we have a binary search tree for each convex hull of height $O(\log n)$~\cite{Overmars-1981}. Therefore, we can compute all intersection points in $O(\log^2 n)$ time.
\end{proof}

We first compute the intersection points of the boundaries of $\ch_p$ and $\ch_h$
in $O(\log^2 n)$ time by Lemma~\ref{lem:intersection}, 
and then sort them along the boundary of their convex hull in clockwise order
in $O(\log n\log\log n)$ time.
Note that this order is the same as the clockwise order along the boundary
of $\ch_p$ (and $\ch_h$).
Then we locate each intersection point on the boundary of each convex hull with respect to the bridges
 in $O(\log n)$ time in total.

There are $O(\log n)$ edges of the convex hull $\ch$ of $\ch_p$ and $\ch_h$ 
that do not appear on the boundaries of $\ch_p$ and $\ch_h$. To distinguish them with the bridges on the boundaries
of $\ch_p$ and $\ch_h$, we call the edges on the boundary of $\ch$ appearing neither $\ch_p$
nor $\ch_h$ the \emph{hull-bridges}. Also we call the  
bridges on $\ch_p$ and $\ch_h$ with endpoints in two distinct cells of $\mathcal{V}$ the \emph{node-bridges}.

The boundary of the convex hull of $\ch_p$ and $\ch_h$ consists of $O(\log n)$ hull-bridges and $O(\log n)$ 
polygonal curves each of which connects two hull-bridges along $\ch_p$ or $\ch_h$. 
We compute all hull-bridges in $O(\log^2 n)$ time.

\begin{lemma}\label{lem:hull-bridges}
	All hull-bridges can be computed in $O(\log^2 n)$ time in total.
\end{lemma}
\begin{proof}
	Let  $\langle p_1,\ldots, p_m\rangle$ be the sequence of the intersection points of the boundaries of $\ch_p$ and $\ch_h$ sorted along the boundary of the convex hull of the intersection points with $m=O(\log n)$. 
	For an index $i$ with $1\leq i< m$, We use $\ch_p[i]$ and $\ch_h[i]$ to denote the parts of the boundaries
	of $\ch_p$ and $\ch_h$ from $p_i$ to $p_{i+1}$, respectively, in clockwise order along their boundaries.
	
	Every hull-bridge is an outer tangent of $\ch_p[i]$ and $\ch_h[i+1]$ or an outer tangent of $\ch_h[i]$ and $\ch_p[i+1]$
	for an index $1\leq i<m$.
	Therefore, it suffices to compute all outer tangents of $\ch_p[i]$ and $\ch_h[i+1]$ (and $\ch_h[i]$ and $\ch_p[i]$). 
	Note that some of the outer tangents are not hull-bridges, but we can determine whether an outer tangent
	is a hull-bridge or not in constant time by considering the edges of $\ch_p$ and $\ch_h$ incident to the endpoints
	of the outer tangent.
	
	We can compute the outer tangents of $\ch_p[i]$ and $\ch_h[i+1]$ in $O(\log n)$ time using the algorithm in~\cite{Overmars-1981} since we have a binary search
	tree of height $O(\log n)$ representing $\ch_p$ (and $\ch_h$). Therefore, we can compute all hull-bridges in $O(\log^2 n)$ time.
\end{proof}

As a result, we obtain a binary search tree of height $O(\log n)$
representing the convex hull $\ch$ of $P\cap Q$. We can
compute the length of each polygonal curve connecting two hull-bridges in $O(\log n)$ time by
Lemma~\ref{lem:basic-op} and the fact that there are $O(\log n)$ node-bridges lying on $\ch$.
Therefore, we have the following theorem.
\begin{theorem}
  Given a set $P$ of $n$ points and a set $\oset$ of $k$ orientations  in the plane,
  we can construct a data structure with $O(nk^3\log^2 n)$ space in $O(nk^3\log^2 n)$ time
  that allows us to compute the perimeter of the convex hull of 
  $P\cap Q$ in $O(\log^2 n)$ time for any 
  query $\oset$-oriented triangle $Q$.
\end{theorem}

\section{\texorpdfstring{$\oset$}{O}-oriented Polygon Queries for Convex
  Hulls}\label{sec:perimeter}
The data structure in Section~\ref{sec:tri} can be used for more
general queries.  We are given a set $P$ of $n$ points in the plane
and a set $\oset$ of $k$ orientations. Let $Q$ be a query $\oset$-oriented
convex $s$-gon. Since $Q$ is convex, $s$ is at most $2k$.
Assume that we are given the three-layer
grid-like range tree on $P$ with respect to the set $\oset$ including
the axis-parallel orientations.  We want to compute the perimeter of
the convex hull of $P\cap Q$ in $O(s\log^2 n)$ time. 

We draw vertical line segments through the vertices of $Q$ to
subdivide $Q$ into at most $2k$ trapezoids.  We subdivide each
trapezoid further using the horizontal lines passing through its
vertices into at most two triangles and one parallelogram.  
The edges of a triangle and a parallelogram, say $\triangle$, have 
orientations in the set $\oset$ including the axis-parallel orientations.
%\complain{vertical and horizontal orientations?}
Thus, we can compute the convex
hull of $\triangle\cap P$ in $O(\log^2 n)$ time and represent it using
a binary search tree of height $O(\log n)$.  By
Lemma~\ref{lem:overmars}, we can compute the convex hull of the points
contained in each trapezoid in $O(s\log^2 n)$ time in total and
represent them using balanced binary search trees of height
$O(\log n)$.

Let $A_1,\ldots,A_t$ be the trapezoids from the leftmost one to the
rightmost one for $t\leq k$.  We consider the trapezoids one by one
from $A_1$ to $A_t$.  Assume that we have just handled the trapezoid
$A_i$ and we want to handle $A_{i+1}$.  Assume further that we already
have the convex hull $\ch_i$ of the points contained in $A_j$ for all
$j\leq i$.  Since the convex hull of the points in $A_{i+1}$ is
disjoint from $\ch_i$, we can compute $\ch_{i+1}$ in $O(\log n)$ time
using Lemma~\ref{lem:overmars}.  In this way, we can compute the
convex hull of $P\cap Q$ in $O(s\log^2 n)$ time in total.  Moreover,
we can compute its perimeter in the same time as we did before.
If $s$ is a constant
as for the case of $\oset$-oriented triangle queries, %in the previous section,
it takes only $O(\log^2 n)$ time.

\begin{theorem}
  Given a set $P$ of $n$ points in the plane and a set $\oset$ of $k$
  orientations, we can construct a data structure with
  $O(nk^3\log^2 n)$ space in $O(nk^3\log^2 n)$ time that allows us to
  compute the perimeter of the convex hull of $P\cap Q$ in
  $O(s\log^2 n)$ time for any $\oset$-oriented convex 
    $s$-gon.
\end{theorem}

As mentioned in Introduction, our data structure can be used to
improve the algorithm and space requirement by Abrahamsen et
al.~\cite{Abrahamsen-2017}.  They considered the following problem:
Given a set $P$ of $n$ points in the plane, partition $P$ into two
subsets $P_1$ and $P_2$ such that the sum of the perimeters of
$\ch(P_1)$ and $\ch(P_2)$ is minimized, where $\ch(A)$ is the convex hull of a
point set $A$.  They gave an $O(n\log^4 n)$-time algorithm for this
problem using $O(n\log^3 n)$ space.  Using our data structure, we can
improve their running time to $O(n\log^2 n)$ and their space
complexity to $O(n\log^2 n)$.
%for the minimum perimeter-sum partition of $n$ points in the plane.
%\complain{Introduce the problem formally if page space is allowed.?}

\begin{corollary}
  Given a set $P$ of $n$ points in the plane, we can compute a minimum
  perimeter-sum bipartition of $P$ in $O(n\log^2 n)$ time using
  $O(n\log^2 n)$ space.
\end{corollary}

\subsection{\texorpdfstring{$\oset$}{O}-oriented Polygon Queries for the Areas of Convex
	Hulls}
\label{ssec:area}
We can modify our data structure to compute the area of the convex
hull of $P\cap Q$ for any $\oset$-oriented convex polygon query $Q$ 
without increasing the time and space complexities.

The modification of the data structure is on the information stored in 
each node of the grid-like range trees of $P$.  Let $v$ be a node of a
level-3 tree of a grid-like range tree $T$.  Without loss of
generality, we assume that the axis of the level-3 trees of $T$ is 
parallel to the $x$-axis.  Let $u_1$ and $u_2$ be the two children of
$v$ such that $B(u_1)$ lies above $B(u_2)$.  We use $\ch(v)$ to denote
the convex hull of the points contained in $B(v)$.

We store its two bridges and the area of the convex hull of each
polygonal chain of $\ch(v)$ lying between two bridges.  In addition,
we store the area of the convex hull of each polygonal chain
connecting an endpoint $e$ of a bridge of $\ch(v)$ and an endpoint
$e'$ of a bridge of $\ch(p(v))$ for the parent node $p(v)$ of $v$ with
$e,e'\in B(v)$ along the boundary of $\ch(v)$.  We do this for every
endpoint of the bridges of $\ch(v)$ and $\ch(p(v))$ that are contained
in $B(v)$.  Therefore, we have the following operation.

\begin{lemma}\label{lem:basic-op-area}
	Given a node $v$ of a level-3 tree %of the three-layer grid-like range tree 
	and two vertices $x, y$ of
	$\ch(v)$, we can compute the area of the convex hull of the part of
	the boundary of $\ch(v)$ from $x$ to $y$ in clockwise order along
	the boundary of $\ch(v)$ in $O(\log n)$ time.
\end{lemma}
\begin{proof}
	The proof is similar to the proof of Lemma~\ref{lem:basic-op}.  We
	obtain two paths such that the boundary of $\ch(v)$ is decomposed
	into $O(\log n)$ pieces each of which corresponds to a node in the
	two paths in $O(\log n)$ time.
	
	In the proof of Lemma~\ref{lem:basic-op}, we simply add the lengths
	of all such pieces, each in $O(\log n)$ time, since the length
	of each piece is stored in a node of the paths.  Instead, we add the
	areas of the convex hulls of all such pieces. Then we compute the
	area of the convex hull of the endpoints of all such pieces. Since
	all such convex hulls are pairwise interior disjoint,
	the total sum is the area we want to compute. Therefore, we can compute the
	area of the convex hull of the part of
	the boundary of $\ch(v)$ from $x$ to $y$ in $O(\log n)$ time.
\end{proof}

For an $\oset$-oriented convex $s$-gon query $Q$, we can obtain
$O(s\log n)$ pieces  
of the convex hull of $P\cap Q$ each of which is a
straight line, or a polygonal curve lying on the boundary of $\ch(v)$
for some node $v$ of the grid-like range trees in $O(s\log^2 n)$ time 
due to Section~\ref{sec:perimeter}. These pieces are sorted along the
boundary of $\ch$.
By construction, the convex hulls of all such pieces
and the convex hull of all straight lines contain the convex hull of $P\cap Q$
and they are pairwise interior disjoint.
The area of the convex hull of $P\cap Q$ is the sum of
the areas of these convex hulls.
Therefore, we can compute the area of the
convex hull of $P\cap Q$ in $O(s\log^2 n)$ time.

%We also have the following results
%for the cases of computing the area and reporting the convex hull of points
%in a query polygon. 
%
%The details are given in Section~\ref{ssec:area} and Section~\ref{ssec:report}
%in Appendix.

\begin{theorem}
	Given a set $P$ of $n$ points and a set $\oset$ of $k$
	orientations in the plane, we can construct a data structure with
	$O(nk^3\log^2 n)$ space in $O(nk^3\log^2 n)$ time that allows us to
	compute the area of the convex hull of $P\cap Q$ in $O(s\log^2 n)$ time for any
	$\oset$-oriented convex $s$-gon.
\end{theorem}

\subsection{\texorpdfstring{$\oset$}{O}-oriented Polygon Queries for Reporting Convex
	Hulls}
\label{ssec:report}
We can use our data structure to report the edges of the convex hull
of $P\cap Q$ for any $\oset$-oriented convex polygon query $Q$ without
increasing the space and time complexities, except an additional $O(h)$ term
in the query time for reporting the convex hull with $h$ edges,
due to the following lemma.
% \complain{The time complexity has additional term $O(h)$, where $h$ is the number
%  of edges of the convex hull.}
\begin{lemma}\label{lem:basic-op-report}
	Given a node $v$ of a level-3 tree and two vertices $x, y$ of
	$\ch(v)$, we can report all edges of the convex hull of the part of
	the boundary of $\ch(v)$ from $x$ to $y$ in clockwise order along
	its boundary in $O(\log n+h(v))$ time, where $h(v)$ is the number of
	the edges reported.
\end{lemma}
\begin{proof}
	The proof is similar to the proof of Lemma~\ref{lem:basic-op}.  We
	obtain two paths such that the boundary of $\ch(v)$ is decomposed
	into $O(\log n)$ pieces each of which corresponds to a node in the
	two paths in $O(\log n)$ time.  For each such node $v$, we can
	report all edges of a part of the boundary of $\ch(v)$ in order once
	we have the endpoints of the part in time linear to the output size by
	traversing the subtree rooted at $v$ in a DFS order.  Therefore, we
	can report all edges of $\ch(v)$ from $x$ to $y$ in $O(\log n+h(v))$
	time in total.
\end{proof}

For an $\oset$-oriented convex $s$-gon query $Q$,
we can decompose the
boundary of $\ch(v)$ into $O(\log n)$ pieces each of which is a
straight line, or a polygonal curve lying on the boundary of $\ch(v)$
for some node $v$ of the grid-like range trees in $O(\log^2 n)$ time
due to Section~\ref{sec:perimeter}.  Using
Lemma~\ref{lem:basic-op-report}, we report the edges of the convex
hull of $P\cap Q$ in $O(s\log^2 n+h)$ time, where $h$ is the number of
the edges of the convex hull.
\begin{theorem}
	Given a set $P$ of $n$ points and a set $\oset$ of $k$
	orientations in the plane, we can construct a data structure with
	$O(nk^3\log^2 n)$ space in $O(nk^3\log^2 n)$ time that allows us to
	report all edges of the convex hull of $P\cap Q$ in $O(s\log^2 n+h)$
	time for any $\oset$-oriented convex $s$-gon, where $h$ is the
	number of edges of the convex hull.
\end{theorem}

\bibliographystyle{plainurl}
\bibliography{paper}{}
\end{document}